\documentclass[10pt]{article}
\usepackage{times}
\usepackage{balance}

\usepackage{subfig}
\usepackage{graphicx}

\usepackage{color}
\usepackage{ulem}

\usepackage{amsmath,amssymb,amsthm}

\newtheorem{theorem}{Theorem}
\newtheorem{lemma}{Lemma}

\begin{document}

\title{Secure Fast Fourier Transform using Fully Homomorphic Encryption}

\author{Thomas Shortell \and Ali Shokoufandeh}

\date{Drexel University \\
  3141 Chestnut St. \\
  Philadelphia, PA 19104 \\
  tms38@drexel.edu,ashokouf@cs.drexel.edu
}

\maketitle

\begin{abstract}
  
Secure signal processing is becoming a de facto model for preserving
privacy. We propose a model based on the Fully Homomorphic Encryption
(FHE) technique to mitigate security breaches.  Our framework provides
a method to perform a Fast Fourier Transform (FFT) on a user-specified
signal.  Using encryption of individual binary values and FHE
operations over addition and multiplication, we enable a user to
perform the FFT in a fixed point fractional representation in binary.
Our approach bounds the error of the implementation to enable
user-selectable parameters based on the specific application.  We
verified our framework against test cases for one dimensional signals
and images (two dimensional signals).

\end{abstract}

\section{Introduction}
\label{sec:intro}

Security breaches are severe situations that can cause significant
problems when using cloud computing resources.  This can occur because
unencrypted data stored within these resources is vulnerable to
security attacks, even if the cloud computing resource is trusted.
Encrypting the data can mitigate such potential vulnerabilities.
However, if the resources are being used to perform significant
computations, then encrypting the data is not normally a possibility.
Our focus is on solving the problem by the ability to process data
while encrypted using Fully Homomorphic Encryption
(FHE)~\cite{gentry:computing}.  FHE enables a user to encrypt their
data and run a prescribed process against the encrypted data.  In this
paper we will focus on secure signal processing particularly, with
performing the Fast Fourier Transform (FFT) using the FHE framework.

Focusing on the overall problem, our user has a set of data that needs
to be processed on a cloud computing resource.  As illustrated in
Figure~\ref{fig:cloudcomputing}, the FHE process involves a user
(client side) using a cloud computing resource (server side).  The
first step in the process is to generate keys for the encryption
((public, secret) key pair).  With this key pair, the original signal
can be encrypted (step 2).  The next step (3) in the process is to
transport the data from the client to server; which is not a major
focus of this paper, except the data should be transmitted via a secure
channel to minimize exposure.  Next, the encrypted process (step 4)
can occur (FFT in this paper).  Here it is important to note that the
processing on the server is actually developed by the user and
transported to the server (this is not shown in the diagram).  Similar
to step 3, step 5 involves getting the data from the server back to
the client (assumed secure connection).  Finally (step 6), the
encrypted processed signal can be decrypted. The decryption results in
the processed signal for the user.
\begin{figure}[ht]
  \centering
  \includegraphics[scale=0.5]{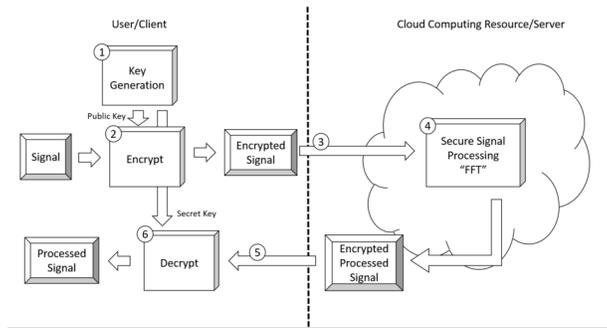}
  \caption{FHE in Cloud Computing: Steps to perform FHE on a cloud computing resources from a data perspective.}
  \label{fig:cloudcomputing}
\end{figure}
This process also leaves open the possibility of generating algorithms that are unknown to the client side.
A complex algorithm that uses an FFT would need the building block of an encrypted FFT.

Starting with the approach used by Shortell and Shokoufandeh \cite{shortell215secruresignal}, we modify it 
to perform an FFT in the encrypted
space.  Our approach focuses on using single binary digit encryption
of fixed point values.  This requires development of binary gate
processors to take binary digit computations to full byte and word
processing.  Once this is done, it is possible to build an encrypted
version of the FFT.  This is analogous to building a CPU and having a
computer program that performs the FFT via additions, subtractions,
and multiplications.

\begin{figure}
  \centering
  \subfloat{
    \includegraphics[scale=1.0]{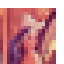}}
  \hfill
  \subfloat{
    \includegraphics[scale=1.0]{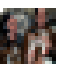}}
  \hfill
  \subfloat{
    \includegraphics[scale=1.0]{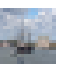}}
  \hfill
  \subfloat{
    \includegraphics[scale=1.0]{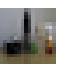}}
  \hfill
  \subfloat{
    \includegraphics[scale=1.5]{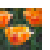}}
  \hfill
  \subfloat{
    \includegraphics[scale=3.0]{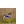}}
  \hfill
  \subfloat{
    \includegraphics[scale=1.0]{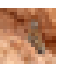}}
  \hfill
  \subfloat{
    \includegraphics[scale=1.0]{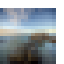}}
  \hfill
  \subfloat{
    \includegraphics[scale=1.0]{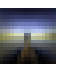}}
  \hfill
  \subfloat{
    \includegraphics[scale=1.0]{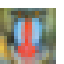}}

  \subfloat{
    \includegraphics[scale=1.0]{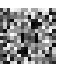}}
  \hfill
  \subfloat{
    \includegraphics[scale=1.0]{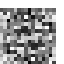}}
  \hfill
  \subfloat{
    \includegraphics[scale=1.0]{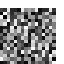}}
  \hfill
  \subfloat{
    \includegraphics[scale=1.0]{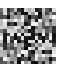}}
  \hfill
  \subfloat{
    \includegraphics[scale=1.0]{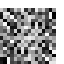}}
  \hfill
  \subfloat{
    \includegraphics[scale=1.0]{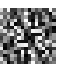}}
  \hfill
  \subfloat{
    \includegraphics[scale=1.0]{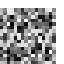}}
  \hfill
  \subfloat{
    \includegraphics[scale=1.0]{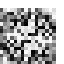}}
  \hfill
  \subfloat{
    \includegraphics[scale=1.0]{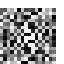}}
  \hfill
  \subfloat{
    \includegraphics[scale=1.0]{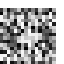}}

  \subfloat{
    \includegraphics[scale=1.0]{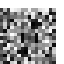}}
  \hfill
  \subfloat{
    \includegraphics[scale=1.0]{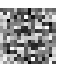}}
  \hfill
  \subfloat{
    \includegraphics[scale=1.0]{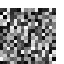}}
  \hfill
  \subfloat{
    \includegraphics[scale=1.0]{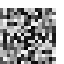}}
  \hfill
  \subfloat{
    \includegraphics[scale=1.0]{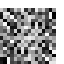}}
  \hfill
  \subfloat{
    \includegraphics[scale=1.0]{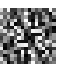}}
  \hfill
  \subfloat{
    \includegraphics[scale=1.0]{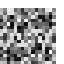}}
  \hfill
  \subfloat{
    \includegraphics[scale=1.0]{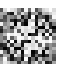}}
  \hfill
  \subfloat{
    \includegraphics[scale=1.0]{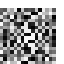}}
  \hfill
  \subfloat{
    \includegraphics[scale=1.0]{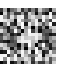}}

  \subfloat{
    \includegraphics[scale=1.0]{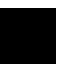}}
  \hfill
  \subfloat{
    \includegraphics[scale=1.0]{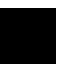}}
  \hfill
  \subfloat{
    \includegraphics[scale=1.0]{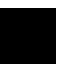}}
  \hfill
  \subfloat{
    \includegraphics[scale=1.0]{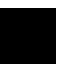}}
  \hfill
  \subfloat{
    \includegraphics[scale=1.0]{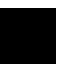}}
  \hfill
  \subfloat{
    \includegraphics[scale=1.0]{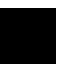}}
  \hfill
  \subfloat{
    \includegraphics[scale=1.0]{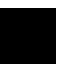}}
  \hfill
  \subfloat{
    \includegraphics[scale=1.0]{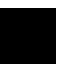}}
  \hfill
  \subfloat{
    \includegraphics[scale=1.0]{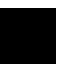}}
  \hfill
  \subfloat{
    \includegraphics[scale=1.0]{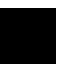}}

  \caption{Example of running FFT in encrypted domain for 2D
    images. The first row shows original images, the second row is the
    unencrypted FFT, third row is the FFT over FHE, and final row is
    the difference of the FFT results. }
  \label{fig:images}
\end{figure}

The remainder of this paper will flow as follows.
Section~\ref{sec:related} presents the related work to encryption and
performing secure signal processing.  We discuss some background
related to the FHE scheme used in Section~\ref{sec:theory}.
Section~\ref{sec:fftinfhe} provides the detailed discussion on
implementation of FFT.  In Sections~\ref{sec:erroranalysis}
and~\ref{sec:experimentalresults} we discuss the accuracy errors and
our implementation evaluation results.
We provide details on Time/Space Complexities of FFT over FHE in Section~\ref{sec:timespacecomplexities}.
Finally, we present our conclusions and future work in Section~\ref{sec:conclusion}.

\section{Related Work}
\label{sec:related}

Research in secure signal processing has received increasing attention
over the past decade.  Troncoso-Pastoriza and Perez-Gonzalez
\cite{troncoso2013secure} examined secure signal processing in the
cloud very similar to the concept we use in our work.  They focused on
privacy issues that occur with cloud computing which is ideally what a
Fully Homomorphic Encryption scheme can provide.  Wang
et. al. \cite{wang2012theoretical} also considered privacy issues with
secure signal processing.  Their focus was on biometrics and
protecting authentication and privacy.  This paper is similar in
keeping confidentiality of private data in a cloud computing
environment.

Other research in this field focuses on using the Paillier encryption
scheme that provides homomorphic addition and constant multiplication.
But the lack of ciphertext-ciphertext multiplication precludes this
scheme from being fully homomorphic. Hsu, Lu, and Pei
\cite{hsu2011homomorphic} used this scheme to perform the Scale
Invariant Feature Transform (SIFT) in an encryption fashion.  Bai
et. al. \cite{bai2014surf} also used the scheme for an encrypted SURF.
There are a few other examples of using the Paillier scheme
\cite{lathey2013homomorphic} \cite{mohanty2013scale}.  While all of
these examples can be performed in the encrypted domain, they fail to be a
Fully Homomorphic Encryption.  In contrast, our approach uses a Fully
Homomorphic Encryption scheme.

It is also important to note some of the recent advances in Fully
Homomorphic Encryption in the past few years.  Gentry developed the
original scheme in
2009~\cite{gentry:computing}~\cite{gentry2009fully}.  The original
scheme was designed for encrypting binary values and improvements over
time continued to look at time and space complexity and the ability to
encrypt more than just binary values 
\cite{brakerski2012leveled} \cite{brakerski2011efficient}.  An
improved scheme \cite{gentry2013homomorphic} developed a few years
later is used in this paper as a FHE tool of choice for our solution

\section{Notation and Background}
\label{sec:theory}

In this paper we use small caps to identify individual binary gates.
$Z_q$ is used to represent an integer ring with a modulus of $q$.
Letters are used for variables in error analysis equations.  For
Complex variables, we will use $\mbox{Re}$ and $\mbox{Im}$ to
represent the real and imaginary values of a complex number.

To perform an encrypted FFT, we need the ability to encrypt and then
process the ciphertexts.  We use the Fully Homomorphic Encryption
scheme defined proposed by Gentry, Sahai, and
Waters~\cite{gentry2013homomorphic} that provides the basic
capabilities including: key generation, encryption, decryption, and
evaluation.  Key generation provides a public/secret key pair that
encrypts with the public key and decrypts with the secret key.
Hardness of their scheme is based on learning with errors
problem~\cite{regev2009lattices}; hence the secret key is a trapdoor
in the public key to extract the original plaintext.  Encryption and
decryption are relatively straightforward conceptual processes.  It is
important to note that the ciphertexts are matrices that embed
integers numbers, and the scheme can operate under the $Z_q$ ring.
Moving into evaluating ciphertexts, the scheme has four capabilities:
addition, constant multiplication, two ciphertext multiplication, and
a \textsc{nand} gate.  The first three operate over the $Z_q$ ring, but last
operation is only for binary values.  We also rely on the fact the
\textsc{nand} gates can be combined into any other gate.

\section{FFT in FHE}
\label{sec:fftinfhe}

To implement FFT in the chosen FHE scheme, it is necessary to develop
a structure that enables using the binary \textsc{nand} gate to perform
encrypted additions and constant multiplications.  Since the
individual ciphertexts are binary values, the additions and
multiplications are going to need to be binary gates.  Additionally,
we need to perform calculations with fractional numbers
given an arbitrary bit size.  So the addition and multiplication
processes in binary need to account for this.

Binary addition and multiplication are computed using different binary gates mainly
\textsc{xor} and \textsc{and} gates.
\textsc{nand} gates are extremely useful because all other binary gates can be calculated from them!
This includes \textsc{and}, \textsc{or}, \textsc{xor}, \textsc{not},
\textsc{nor}, and \textsc{xnor}.  Having these gates available enables
generation of more complex binary processes.
Many of these gate computations are well known techniques.

Our next step is to use these binary gates and generate half and full
adders.  Having half and full adders will allow for performing an
arbitrary bit size addition of two ciphertexts.  Binary addition (and
subtraction) is relatively straightforward.  Full adders linked to
together starting from the lowest bit to the highest bit will generate
the addition of two ciphertexts.  Figure~\ref{fig:eightbitadder} shows
how the adders are built together to perform arbitrary bit size
addition (example of eight bits).  Interestingly, subtraction can be
performed by inverting the second binary values and using an initial carry bit
of 1.

\begin{figure}[ht]
  \centering
  \includegraphics{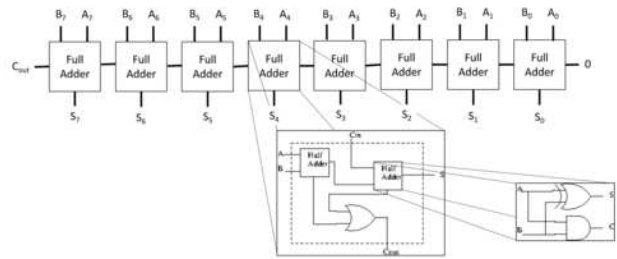}
  \caption{Example Adder for 8 bits including half and full adders}
  \label{fig:eightbitadder}
\end{figure}

Next, we focus on constant multiplication and ciphertext-ciphertext multiplication.
Binary multiplication causes the doubling of size for the binary values, which is an extremely important fact because of a two's complement implementation of binary numbers.
Bit extension is necessary because otherwise negative numbers will not be calculated correctly.
We used an implementation known as Wallace Trees for multiplication \cite{wallace1964suggestion} as it provides a way to arbitrarily handle different binary sizes \footnote{for brevity, we refer the reader to \cite{wallace1964suggestion} for implementation}.
There is a potential security concern that may be introduced by constant multiplication.
Because the constant is unencrypted, the binary values of the original constant can be partially traced into the Wallace tree computations.
This can potentially provide an attacker knowledge about intermediary values.
While this will not review the original ciphertext, it enables the attacker to identify values that the encrypted result cannot be.
For example, if a few zeroes in the constant factor can identify a result bit to be zero, the space of potential values of the result has been reduced.
Technically this is true of just multiplying an even constant factor because the result must be even.

At this point, we have methods to add, subtract, and multiply
integers. The FFT algorithm will require computations involving
floating point numbers.  Our implementation uses fixed point numbers
to emulate floating point numbers.  By using a fixed point
representation, we can use the integer based numbers represent
floating point numbers.  This was similar to what was done in other
work \cite{shortell215secruresignal} but we improve this by using
binary digits to handle cases that the straight integer implementation
can't handle.  Using a multiplier that is a factor of 2 (based on
using binary values), floating point numbers can become fixed point
numbers in an integer space.  This approach introduces error which we
will discuss in the next section.  For addition and subtraction, fixed
point implementation is easily provided that the fractional bit size
is consistent for both numbers (which we enforce).  Multiplication is
complicated because of the expansion in the bit size (doubling) and
our solution is to extract a different set of the bits than just the
lowest $x$ (where $x$ is the input bit size).  By extracting the
middle bits of the input size, an implicit division is occurring which
provides a true implementation of fixed point multiplication.  This
also solves an earlier problem that caused long term expansion of
fixed point multipliers during multiplication.

We have all the necessary building blocks required to implement FFT in the
encrypted space.  Step one of FFT is to perform a bit reversal, which
is extremely easy since the order of the points is known and can
easily be modified in position without revealing anything other than a
bit reversal occurred.  Step two requires calculating the $W_n$
multipliers and running the butterfly computations over the signal for
$\log N$ iterations ($N \log N$ driver).  The multipliers are
constants, so these can be input and used with constant
multiplication.  Constant multiplication results drive the additions
for the signal points; yielding the two new points for the next
iteration of the FFT.  After $N$ iterations of signal calculations and
$\log N$ iterations of the outer loop, the process will be complete.

\begin{theorem}
  \label{thm:fftinfhe}
  Fast Fourier Transform can be calculated in the encrypted domain via
  Fully Homomorphic Encryption.
\end{theorem}

To prove Theorem~\ref{thm:fftinfhe}, we need to show that the process
does in fact perform the FFT in the encrypted domain. Following the
process, we must prove that the encryption works, the evaluation
piece, and then the decryption piece.

\begin{lemma}
  \label{lemma:encryption}
  The FHE scheme properly encrypts values in the fixed fractional
  format to a vector of encrypted ciphertexts.
\end{lemma}

\begin{lemma}
  \label{lemma:decryption}
  The FHE scheme properly decrypts values in the fixed fractional
  format from a vector of ciphertexts.
\end{lemma}

\begin{lemma}
  \label{lemma:evaluation}
  The FHE scheme properly evaluates the Fast Fourier Transform for binary vector ciphertexts of a vector of fixed fractional values.
\end{lemma}

To prove Lemma~\ref{lemma:evaluation}, we need to show that the output
of each binary vector is the same result as expected by the FFT
algorithm.  The key point here is that the butterfly computations of
FFT provide the correct result.  As we had shown earlier, addition,
subtraction, and multiplication are needed to perform the butterfly
computation.  This becomes recursive as each operation is built from
binary gates and finally at the \textsc{nand} gate from the scheme itself.  We
start with individual lemmas for the individual gates and move up to
three main operations.  We assume the \textsc{nand} gate is working as
part of these lemmas.  Theorem 3 of \cite{gentry2013homomorphic}
proves the working of the FHE scheme basics (include Lemmas~\ref{lemma:encryption} and~\ref{lemma:decryption}).

\begin{figure}
  \centering
  \subfloat[\textsc{not}]{
    \includegraphics[width=.4\linewidth]{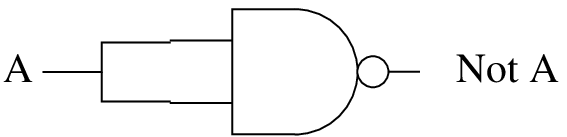}\label{fig:gates:not}}
  \hfill
  \subfloat[\textsc{and}]{
    \includegraphics[width=.4\linewidth]{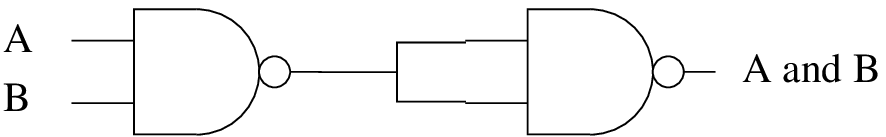}\label{fig:gates:and}}\\
  
  \subfloat[\textsc{or}]{
    \includegraphics[width=.4\linewidth]{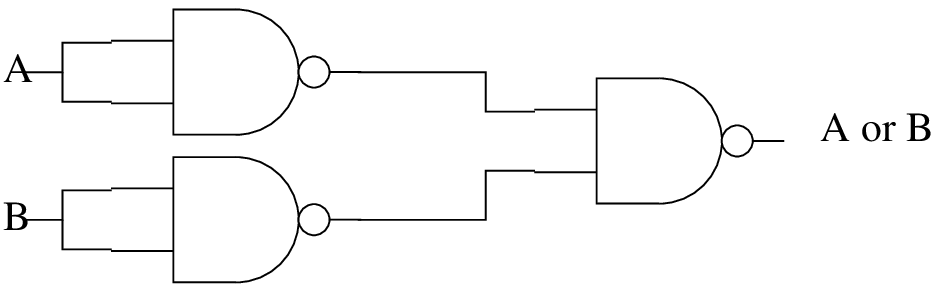}\label{fig:gates:or}}
  \hfill
  \subfloat[\textsc{nor}]{
    \includegraphics[width=.4\linewidth]{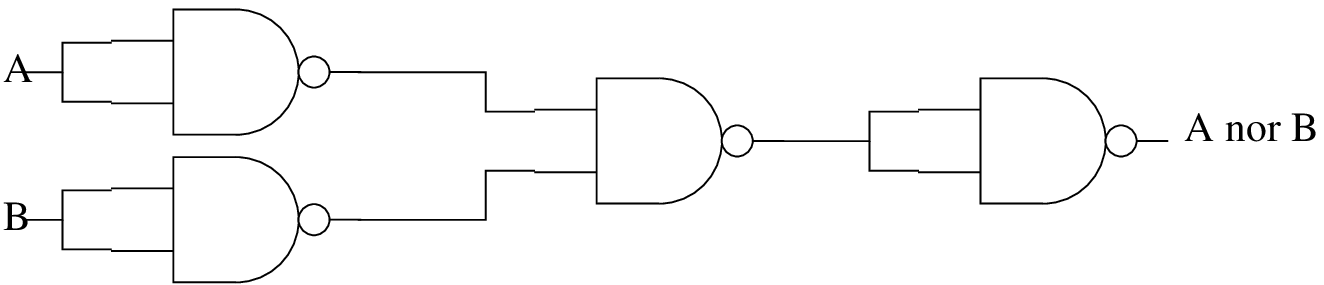}\label{fig:gates:nor}}\\
  
  \subfloat[\textsc{xor}]{
    \includegraphics[width=.4\linewidth]{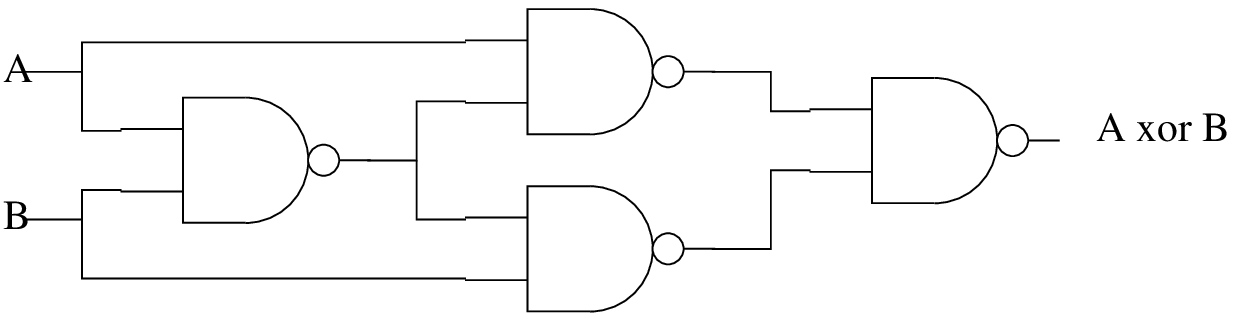}\label{fig:gates:xor}}
  \hfill
  \subfloat[\textsc{xnor}]{
    \includegraphics[width=.4\linewidth]{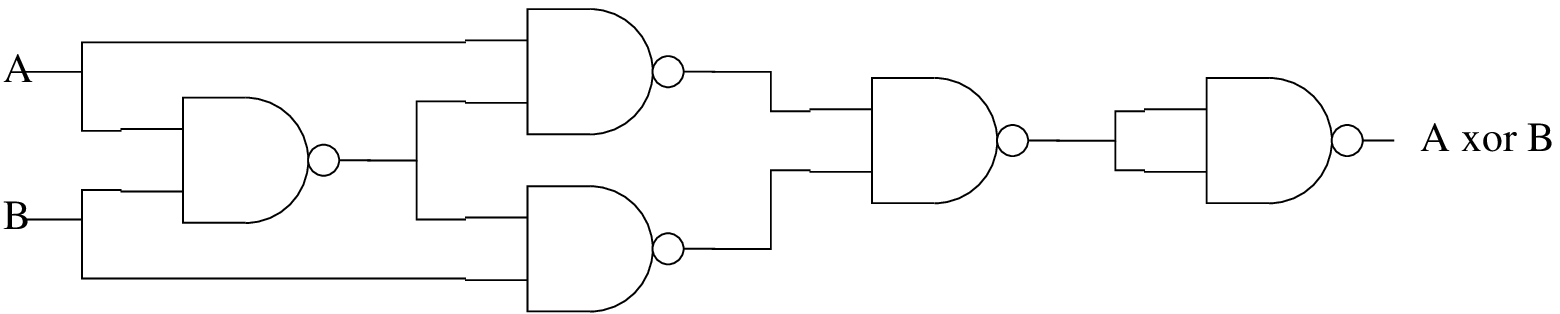}\label{fig:gates:xnor}}
  
  \caption{Binary Gate Implementations using only \textsc{nand} gates}
  \label{fig:binarygates}
\end{figure}

\begin{lemma}
  \label{lemma:binarygates}
  Given the \textsc{nand} gate of the FHE scheme works, the FHE scheme properly performs a \textsc{not}, \textsc{and}, \textsc{or}, and \textsc{xor}  gates. 
\end{lemma}\footnote{We omit \textsc{nor} and \textsc{xnor} because they are not used for FFT.}

\begin{proof}
  Previous research \cite{wesselkamper1975sole} has shown that a NAND gate can be used to generate all other logic gates including \textsc{not}, \textsc{and}, \textsc{or}, and \textsc{xor} (which are important to this paper).
  Examples of the logic gates are shown in Figure~\ref{fig:binarygates}.
  Because the FHE scheme can compute a \textsc{nand} gate,  therefore we know the framework will correctly compute the remaining logic gates. 
\end{proof}

With Lemmas for the individual binary gates, we can create Lemmas for the half and full adders.
Having the individual Lemmas will make the proofs significantly easier. 

\begin{lemma}
  \label{lemma:halfadder}
  Given that a \textsc{xor} and \textsc{and} gate are available in the FHE scheme, the FHE scheme properly calculates a half adder. 
\end{lemma}

\begin{proof}
  From Figure~\ref{fig:eightbitadder}, we have the half adder structure in the bottom right hand corner.
  We are not proving the half adder itself is right, but proving that it is correctly calculated by the FHE scheme via the gates, which is true because of the Lemmas~\ref{lemma:binarygates}.  
\end{proof}

\begin{lemma}
  \label{lemma:fulladder}
  Given that an \textsc{or} gate and a half adder are available in the FHE scheme, the FHE scheme properly calculates the full adder. 
\end{lemma}

\begin{proof}
  Reviewing the inset in Figure~\ref{fig:eightbitadder}, we know the full adder is built from two half adders and an \textsc{or} gate (again not proving a full adder structure is correct).
  Lemma~\ref{lemma:binarygates} proved the \textsc{or} gate works.
  Lemma~\ref{lemma:halfadder} proved a single half adder works.
  By construction, following the correct paths provides a full adder capability. 
\end{proof}
It is now time to move on to addition and subtraction.
\begin{lemma}
  \label{lemma:addsubtract}
  Given the binary gate and adder capabilities of the FHE scheme, the
  FHE scheme properly calculates arbitrary bit addition and
  subtraction of integers and fixed fractional format.
\end{lemma}

\begin{proof}
  Proving the scheme can correctly calculate addition and subtraction requires proving the binary result is correct.
  Since the algorithm is the same for both integers and the fixed fractional format, proving one proves both.
  We have the basic algorithm for adding and subtracting binary numbers and it is know to comprise a set of full adders.
  Lemma~\ref{lemma:fulladder} proves the scheme can calculate a full adder correctly.
  Given this, we have proved that the addition and subtraction of ciphertexts occurs correctly in our FHE scheme. 
\end{proof}

Our next two lemmas are related to multiplication. Our first lemma
will involve proving multiplication works for simple integer numbers
including truncating the high order bits.  The second lemma focuses on
fixed fractional format and obtaining the correct resulting value.  An
interesting note at this point because of the lemma is that binary
division could actually be implemented but was not because it was not
needed for FFT; however a single division is needed to support the
fixed fractional format.

\begin{lemma}
  \label{lemma:integermultiplication}
  Given \textsc{and} gates and half and full adders, the FHE scheme
  properly calculates integer multiplication via Wallace Trees.
\end{lemma}

\begin{proof}
  Starting off, we know that the Wallace Trees correctly calculate integer multiplication of binary values.
  To prove our scheme calculates multiplication of two binary integers correctly, we need to show that we can correctly compute the Wallace Trees.
  Wallace Trees are built from a set of \textsc{and} gates, followed by combinations of half and full adders.
  All three of these correctly compute binary values in our scheme based on our previous lemmas.
  The final importance of the proof is to show both positive and negative can be correctly calculated.
  This can be done easily by bit extending the binary values to twice the original size.
  Finally, we drop the higher order bits to keep the binary values size at their original size.
  Therefore, our scheme correctly calculates integer multiplication in the encrypted domain. 
\end{proof}

\begin{lemma}
  \label{lemma:fixedmultiplication}
  Given \textsc{and} gates and half and full adders, the FHE scheme
  properly calculates fixed fractional format multiplication via
  Wallace Trees.
\end{lemma}

\begin{proof}
  We start with the fact that we can assume to have \textsc{and} gate, half adder and a full adder that correctly compute there binary values in FHE.
  From Lemma~\ref{lemma:integermultiplication}, our scheme will properly calculate integer multiplication, which is part of our fixed fractional format.
  In a fixed number format, we need to divide the scaling factor out of the multiplied values, which is $2n$ of the previous size.
  Because we forced the multiplier to be a power of two, division is a simple bit shift.
  This is easy in our FHE scheme to drop the lower and higher order bits (extract the middle bits).
  Thus, the scheme calculates the fixed fractional format multiplication. 
\end{proof}

Using the above lemmas and their respective proofs, we can prove the
evaluation lemma (\ref{lemma:evaluation}) which is the key butterfly
computation of the FFT; without which the correctness of computations
for FFT in encryption domain cannot be proven to be correct.

\begin{proof}
  Remembering the two parts of the FFT, first is the bit reversal and then the butterfly computations.
  The bit reversal does not need to occur in the encrypted domain so this process works correctly as before.
  As the butterfly computations are processed in a loop and the loop processing is not encrypted, proving that the butterfly computations can be calculated in the encrypted domain will prove that the FFT can be calculated using FHE. 

  Simply, the butterfly computation is a pair of equations of complex numbers.
  Equivalently, there are six real number equations to calculate.
  These equations are a combination of addition, subtraction, and multiplication (via constant values).
  Lemmas~\ref{lemma:addsubtract} and~\ref{lemma:fixedmultiplication} proved that the FHE scheme correctly calculates addition, subtraction, and multiplication of encrypted fixed fractional values.
  Because of this, we can calculate the butterfly computations and properly calculate the FFT.
\end{proof}

Finally we can prove Theorem~\ref{thm:fftinfhe}.
\begin{proof}
  To prove that FFT can be calculated via an FHE method, we focused on
  three pieces in the FHE scheme: Encryption, Evaluation, and
  Decryption.  We have used three individual lemmas to prove that each
  one of these pieces is correctly computed in FHE
  (Lemmas~\ref{lemma:encryption},~\ref{lemma:evaluation},
  and~\ref{lemma:decryption}).  Because we have proved the entire
  process works, we have proved the FFT can be computed in the
  encrypted domain using our FHE scheme.
\end{proof}

\section{Error Analysis}
\label{sec:erroranalysis}

Having discussed the implementation of FFT using FHE, we turn our
attention to estimating the error. Our only source of error is the
conversion of floating point numbers to fixed point numbers.  While
initially this might not be a major problem, over time the
calculations will lose accuracy.  Immediately, we know that there will
be dependencies on the number of points in the signal because of the
dependency on the size with the number of iterations.

We start with building up the bounds of the error introduced by the
fixed point representation.  There will some initial error caused by
truncation when moving to fixed point representation.  Following that,
the main computation is the two point butterfly, which involves a
complex point multiplication and an addition.  We complete this
section with a proof on the error bound for the entire FFT in FHE.

\begin{lemma}
  \label{lemma:cpmult}
  For a multiplication involving two complex points ($a+bi$, $c+di$) with an
  error of $\Delta,\ (0 \leq \Delta < 1)$ in each of the real and
  imaginary components, the error in the resultant is bounded by:
  \begin{equation}
    \left( \Delta \left( a + c - b - d \right), \Delta \left( a + b + c + d \right) \right).
  \end{equation}
\end{lemma}

\begin{proof}
  Our two initial points are:
  \begin{eqnarray}
    (a + \Delta, b + \Delta) \\
    (c + \Delta, d + \Delta) 
  \end{eqnarray}
  Multiplying these together yields,
  \begin{eqnarray}
    \left( \left( a + \Delta \right) \cdot \left( c + \Delta \right) - \left( b + \Delta \right) \cdot \left( d + \Delta \right), \right. \nonumber\\
    \left. \left( b + \Delta \right) \cdot \left( c + \Delta \right) + \left( a + \Delta \right) \cdot \left( d + \Delta \right) \right)
  \end{eqnarray}
  As we expand the factors, we will drop terms of $\Delta^2$ as these will not be the primary source of error in our final equations (because $\Delta < 1$).
  \begin{eqnarray}
    \left( a \cdot d + \Delta \cdot ( a + c ) - \left( b \cdot d + \Delta \cdot ( b + d ) \right), \right. \nonumber\\
    \left.  b \cdot c + \Delta \cdot ( b + c ) + a \cdot d + \Delta \cdot ( a + d ) \right)
  \end{eqnarray}
  Continuing to combine terms,
  \begin{eqnarray}
    \left( a \cdot d - b \cdot d + \Delta \cdot ( a + c - b - d ), \right. \nonumber\\
    \left.  b \cdot c + a \cdot d + \Delta \cdot ( b + c + a + d ) \right)
  \end{eqnarray}
  This provides errors in the real and imaginary components as:
  \begin{eqnarray}
    \left( \Delta \cdot ( a + c - b - d ), \right. \nonumber\\
    \left. \Delta \cdot ( b + c + a + d ) \right)
  \end{eqnarray}
\end{proof}

Next, we examine the error for a single butterfly computation.  The
original FFT equation is:
\begin{equation}
  X_i = x_i + W_n * x_j.
\end{equation}

\begin{lemma}
  \label{lemma:butterfly}
  For a single butterfly computation, the error is bounded by
  \begin{eqnarray}
    \Delta \cdot \left( \mbox{Re}(W_n) + \mbox{Im}(W_n) +
    \mbox{Re}(x_j) + \mbox{Im}(x_j) + 1 \right)
  \end{eqnarray}
  for a single butterfly computation of a single point where $\Delta$ is the original error. 
\end{lemma}

\begin{proof}
  We start with the butterfly computation equation for a single point:
  \begin{equation}
    X_i = x_i + W_n \cdot x_j 
  \end{equation}
  Then adding $\Delta$ to each of the terms:
  \begin{equation}
    X_i = x_i + \Delta + \left( W_n + \Delta \right) \cdot \left( x_j + \Delta \right)
  \end{equation}
  Lemma~\ref{lemma:cpmult} provides the bound on the error for the right hand side of the equation.
  Then there is only addition of another $\Delta$ in both real and imaginary parts.
  This yields:
  \begin{eqnarray}
    \Delta \cdot \left( \mbox{Re}(W_n) + \mbox{Re}(x_j) - \mbox{Im}(W_n) - \mbox{Im}(x_j) + 1 \right) \nonumber\\
    \Delta \cdot \left( \mbox{Re}(W_n) + \mbox{Re}(x_j) + \mbox{Im}(W_n) + \mbox{Im}(x_j) + 1 \right)
  \end{eqnarray}
  On the second half of the butterfly computations, the signs of the real and imaginary values are flipped (the one is always positive). 
\end{proof}

The previous two lemmas help us bound the overall error of the FFT by
estimating the error accumulated over time.  The most important item
about the butterfly computations is that they reduce an $O(N^2)$
algorithm to a $O(N \log N)$ algorithm.  This means the resulting
value the FFT for a single point is a $O(N^2)$-based value that only
uses $O(N \log N)$ computations.  This helps bound the error from
above: the summation of all signal points and the $W_n$ values.  But
because we perform less computations, the value can be bounded lower.

\begin{theorem}
  \label{thm:fftfheerror}
  For performing FFT in FHE, the error introduced by the processing is
  bounded by:
  \begin{equation}
    \Delta \cdot \frac{N}{2} \left( \log N + X_b + 1 \right),
  \end{equation}
  where $\Delta$ is the original representation, $W_S$ (used in proof)
  is a sum of all $W_n$ that appear in the FFT for a given size
  $N$, and $X_b$ is a bound on the size of the signal points.
\end{theorem}

\begin{proof}
  Using Lemma~\ref{lemma:butterfly}, we know after a single butterfly
  for a value the error is:
  \begin{eqnarray}
    \Delta \cdot \left( \mbox{Re}(W_n) + \mbox{Re}(x_j) - \mbox{Im}(W_n) - \mbox{Im}(x_j) + 1 \right), \nonumber\\
    \Delta \cdot \left( \mbox{Re}(W_n) + \mbox{Re}(x_j) + \mbox{Im}(W_n) + \mbox{Im}(x_j) + 1 \right).
  \end{eqnarray}
  After the second iteration of the outer loop of FFT, the $W_n$ will
  be unchanged but the $x_j$s will have additional values (and
  error).  The key point here is what is being added to the error over
  time.  It is actually the real and imaginary components of the
  signal points.  This tells us at the second iteration, we have added
  up all $W_n$  plus a number of terms from each $x_j$ seen
  so far, and an equivalent number of ones from the $x_i$ sides.  After
  the $\log N$ iterations, each point will have the following error
  contributions:
  \begin{equation}
    \Delta \cdot \left( \sum_{W_n \in W} W_n + \sum_{\mbox{even} j} x_j + \frac{N}{2} \right),
  \end{equation}
  where we have used $W$ to represent the set of $W_n$.  Since each $W_n$
  is a known constant for a given $N$, we will call this
  sum bound as $W_S$.  Additionally, we know that we can assume
  a bound on each $x_j$ without loss of generality.  Calling this
  value $X_b$ and knowing there are $\frac{N}{2}$ of them, we can
  update the bound as:
  \begin{equation}
    \Delta \cdot \left( W_S + \frac{N}{2} \left( X_b + 1 \right) \right)
  \end{equation}
  An equivalent view on  $W_S$ parameter is
  that the $W_n$ absolute values are less than one, which means $W_S$
  is bounded by the number of times any $W_n$ enters an equation
  (i.e., $\frac{N \log N}{2}$) resulting in a total error of
  \begin{equation}
    \Delta \cdot \left( \frac{N}{2} \left( \log N + X_b + 1 \right) \right).
  \end{equation}
\end{proof}

This provides some different results than
\cite{shortell215secruresignal}.  In particular, the error dependents
on the total number of points, which in turn means that the error will
increase as the size of the signal increases.  So when choosing a
multiplier with the fixed fractional format, the signal size will
matter. This is a good place to remember that an FHE scheme can evaluate ciphertexts indefinitely by refreshing ciphertexts (see \cite{gentry:computing} and \cite{gentry2013homomorphic} for details).

\section{Experimental Results}
\label{sec:experimentalresults}

To verify our implementation works, we ran random signals against the
FFT in FHE and compared the results to an unencrypted standard version
FFT.  Our main focus is to measure the error introduced by our
implementation and verify Theorem~\ref{thm:fftfheerror} is valid.  We
used signal sizes of 8, 16, 32, 64, and 128 complex data signals.
Additionally, we constrained the signal to be values in the range
$[0,1]$ similar to what many real world signals operate in.  In the
encrypted domain, we work in a 32-bit binary space and 16-bit
fractional space.  This means we limit our non-fractional integer size
to 16-bits and we are using a multiplier of 65536.  Having a
multiplier of 65536 is useful since it will provide numerical accuracy
up to $1/65536 \approx 1.5259e-5$.  As we saw in the previous section,
we will not be containing the entire set of values in this space
because over time the processing will lose accuracy (particularly in
multiplication).

It is important to look at measured error introduced as a whole for
given signals.  Our main focus is to make sure we can constrain the
error.  When comparing individual values between the unencrypted and
encrypted results, we can calculate the total error sum of all points
($2n$ from $n$ complex points).  We also calculate the mean error
across the points along with the variance and standard deviation.
Table~\ref{tbl:measurederror} shows the results of the measured error
from the random testing.  One of the key aspects of the results is
that the average error per point is slowly rising upwards above zero
but is staying in the $10^{-5}$ range.  This is expected as the error
is dependent on the total number of points in the signal.
Considering our accuracy of $1/65536$ and that the signal size is a multiplicative factor as well, experimentally the error is contained below the bounding from the previous section. 

\balance

\begin{table}[ht]
  \scriptsize
  \caption{Measured Error in One Dimensional Tests}
  \centering
  \begin{tabular}{|l|c|c|c|c|}
    \hline
    Complex & Total & Average &  & Standard \\
    Point Size & Error & Error & Variance & Deviation \\
    \hline
    8 pt & 0.0207 & 1.294e-5 & 8.208e-11 & 9.060e-06 \\
    \hline
    16 pt & 0.0709 & 2.216e-5 & 2.434e-10 & 1.560e-05  \\
    \hline
    32 pt & 0.258 & 4.199e-5 & 1.714e-09 & 4.140e-05 \\
    \hline
    64 pt & 1.030 & 8.383e-5 & 4.400e-09 & 6.633e-05 \\
    \hline
    128 pt & 4.437 & 1.81e-4 & 2.856e-08 & 1.69e-04 \\
    \hline
  \end{tabular}
  \label{tbl:measurederror}
\end{table}

Finally, we tested our configuration using $16 \times 16$ images (10
total images; shown in Figure~\ref{fig:images}).  This provided a two
dimensional test of the FFT over FHE in the multidimensional case.  We
calculated the same values in the one dimensional case.  The total
error was $0.311$ which led to an average error of $6.067\times
10^{-5}$.  This equates to a variance of $5.558 \times 10^{-9}$ and a
standard deviation of $7.455 \times 10^{-5}$.  These values show our
framework contains the error within bounds.

\section{Time/Space Complexities}
\label{sec:timespacecomplexities}

FHE schemes are known to be computationally intensive.
Encrypting a single value plaintext generates a two dimensional matrix in the ciphertext space.
A \textsc{nand} gate in the FHE scheme is an $O(N^3)$ process (matrix-matrix multiplication). 
This can be partially reduced by using parallelization techniques (GPU processing is real potential here).
Focusing on FFT, which is an $O( M \log M)$ process, we need to understand the number of gates that can processed in total.
$M \log M$ is the number of addition and multiplication gates being processed in total.
A fixed point addition process is $36*F$ \textsc{nand} gates, where $F$ is the size of fixed point size and $36$ comes from the nand gates in the $F$ sequential full adders.
A fixed point multiplication process is $288 F^2 \log F$ \textsc{nand} gates.
$4F$ \textsc{and} gates for the first step and worst case assumed full adders for the $\log F$ process.
This second process is for $2F$ full adders.
It should be noted that this bound is worst case and can be significantly tighten because Wallace trees do not need to run full adders at each step. 
Combining these calculations yields an asymptotic running time of $O( M \log{M} F^2 \log{F} N^3$ with constants removing, where we have used $M$ to be the number of data points in the FFT, $F$ to be the fixed point binary size, and $N$ to be ciphertext size.

Coming back to the space complexity, we know a ciphertext text has $N^2$ space (matrix).
We will set the size of the signal to be $L$, ie the total amount of data points in the signal whether single or multi-dimensional. 
Each data point has a set of $F$ binary values.
Multiplying these together yield: $O(L F N^2)$ space complexity.
In both cases, these are extremely high complexities (especially compared to unencrypted processing).

\section{Conclusion}
\label{sec:conclusion}

Having shown that we can perform the Fast Fourier Transform in the
encrypted domain, we are now able to expand the capabilities of the
FHE framework to target additional signal processing algorithms and
other potential image processing algorithms like SIFT and SURF.  There
are other open issues with FHE.  One major open item is that the FHE
processing takes significant amount of time because of the
matrix-matrix multiplication required for the underlying processing.
Being able to improve computational performance of the FHE processing
would contribute significantly to making FHE a viable solution in real
world computing.

\bibliographystyle{abbrv}
\bibliography{fftpapernew}

\end{document}